\documentclass[a4paper,UKenglish,cleveref, autoref]{lipics-v2019}

\usepackage{multirow,array,booktabs}
\nolinenumbers

\title{Online Multistage Subset Maximization Problems} 

\author{Evripidis Bampis}{Sorbonne Universit\'e, CNRS, LIP6 UMR 7606, 4 place Jussieu, 75005 Paris, France.}{evripidis.bampis@lip6.fr}{}{}

\author{Bruno Escoffier}{Sorbonne Universit\'e, CNRS, LIP6 UMR 7606, 4 place Jussieu, 75005 Paris, France.}{bruno.escoffier@lip6.fr}{}{}

\author{Kevin Schewior}{Fakult\"at f\"ur Informatik, Technische Universit\"at M\"unchen, Germany;\\ D\'epartment d'Informatique, \'Ecole Normale Sup\'erieure Paris, PSL University, France }{kschewior@gmail.com}{}{Supported by a DAAD PRIME grant.}

\author{Alexandre Teiller}{Sorbonne Universit\'e, CNRS, LIP6 UMR 7606, 4 place Jussieu, 75005 Paris, France.}{alexandre.teiller@lip6.fr}{}{}

\authorrunning{E. Bampis et al.}

\Copyright{Evripidis Bampis and Bruno Escoffier and Kevin Schewior and Alexandre Teiller}



\category{}

\relatedversion{}

\supplement{}

\acknowledgements{This research benefited from the support of FMJH program PGMO and from the support of EDF-Thal\`es-Orange.}

\hideLIPIcs  

\EventEditors{John Q. Open and Joan R. Access}
\EventNoEds{2}
\EventLongTitle{Conference on Very Important Topics (CVIT) 2019}
\EventShortTitle{CVIT 2019}
\EventAcronym{CVIT}
\EventYear{2019}
\EventDate{December 24--27, 2019}
\EventLocation{Little Whinging, United Kingdom}
\EventLogo{}
\SeriesVolume{42}
\ArticleNo{23}

\newcommand\ot{multistage}
\newcommand\easy{Hamming}

\newcommand\hard{Intersection}

\begin{document}

\maketitle

\begin{abstract}
 Numerous combinatorial optimization problems (knapsack,  maximum-weight matching, etc.) can be expressed as \emph{subset maximization problems}: One is given a ground set $N=\{1,\dots,n\}$, a collection $\mathcal{F}\subseteq 2^N$ of subsets thereof such that $\emptyset\in\mathcal{F}$, and an objective (profit) function $p:\mathcal{F}\rightarrow\mathbb{R}_+$. The task is to choose a set $S\in\mathcal{F}$ that maximizes $p(S)$. We consider the \emph{multistage} version (Eisenstat et al., Gupta et al., both ICALP 2014) of such problems: The profit function $p_t$ (and possibly the set of feasible solutions $\mathcal{F}_t$) may change over time. Since in many applications changing the solution is costly, the task becomes to find a sequence of solutions that optimizes the trade-off between good per-time solutions and stable solutions taking into account an additional similarity bonus.
As similarity measure for two consecutive solutions, we consider either the size of the intersection of the two solutions or the difference of $n$ and the Hamming distance between the two characteristic vectors. 

We study multistage subset maximization problems  in the \emph{online} setting, that is, $p_t$ (along with possibly $\mathcal{F}_t$) only arrive one by one and, upon such an arrival, the online algorithm has to output the corresponding solution without knowledge of the future.
 
We develop general techniques for online multistage subset maximization and thereby characterize those models (given by the type of data evolution and the type of similarity measure) that admit a constant-competitive online algorithm. When no constant competitive ratio is possible, we employ lookahead to circumvent this issue. When a constant competitive ratio is possible, we provide almost matching lower and upper bounds on the best achievable one.
\end{abstract}

\newpage

\section{Introduction}

In a classical combinatorial optimization setting, given an instance of a problem one needs to find a good feasible solution. However, in many situations, the data may evolve over the time and one has to solve a sequence of instances. The natural approach of solving every instance independently may induce a significant transition cost, for instance for moving a system from one state to another. This cost may represent e.g. the cost of turning on/off the servers in a data center \cite{LinWAT13,BansalGKPSS15,AntoniadisS17,AlbersQ18}, the cost of changing the quality level in video streaming \cite{Joseph}, or the cost for turning on/off nuclear plants \cite{thesececile}.  
Gupta et al. \cite{Gupta} and Eisenstat et al. \cite{Eisenstat}  proposed
 a \emph{multistage} model where given  a time horizon $t = 1, 2, \ldots,T$,  the input is a sequence of instances $I_1,I_2,\ldots,I_T$, (one for each time step), and the goal is to find a sequence of solutions $S_1,S_2,\ldots,S_T$ (one for each time step) reaching  a tradeoff between the quality of the solutions in each time step and the stability/similarity of the solutions in consecutive time steps. The addition of the transition cost makes some  classic combinatorial optimization problems much harder. This is the case  for instance  
 for the minimum weighted perfect matching problem in the \emph{off-line} case where the whole sequence of instances is known in advance. While the one-step problem is polynomially-time solvable, the multistage problem becomes hard to approximate even for bipartite graphs and for only two time steps \cite{Bampis,Gupta}. 
 
 In this work, we focus on the \emph{on-line}  case, where at time $t$ no knowledge is available for instances at times $t+1, \ldots, T$.  When it is not possible to handle the on-line case, we turn our attention to the $k$-\emph{lookahead} case, where at time $t$ the instances at times $t+1, \ldots, t+k$ are also known. This case is of interest since in some applications like in  dynamic capacity planning in data centers,
 the forecasts of future demands may be very helpful  \cite{Lin,Liu}. Our goal is to measure the impact of the lack of  knowledge of the future  on the quality and the stability of the returned solutions. 
 Indeed, our algorithms are limited in their knowledge of the sequence of instances.
Given that the number of time steps is given, we compute the competitive ratio of the algorithm after time step $T$: As we focus on maximization problems, we say that an algorithm is (strictly) $\alpha$-competitive (with competitive ratio $\alpha$) if its value is at least $\frac1\alpha$ times the optimal value on all instances.

As it is usual in the online setting, we consider no limitations in the computational resources available. This means that at every time step $t$, where instance $I_t$ is known, we assume the existence of an oracle able to compute the optimal solution for that time step.  Notice also that our lower bounds do not rely on any complexity assumptions. 
Some recent results are already known for the on-line multistage model  \cite{Bampis+,Gupta}, however all these results are obtained for specific problems.
In this work, we study multistage variants  of a broad family of maximization problems. The family of optimization problems that we consider is the following.


\begin{definition}

\emph{(Subset Maximization Problems.)} A Subset Maximization problem $\cal P$ is a combinatorial optimization problem whose instances $I=(N,p,\mathcal{F})$ consist of
\begin{itemize}
    \item A ground set $N$;
    \item A set $\mathcal{F}\subseteq 2^N$ of feasible solutions such that $\emptyset\in\mathcal{F}$;
    \item A positive weight $p(S)$ for every $S \in \mathcal{F}$.
\end{itemize}
The goal is to find $S^*\in \mathcal{F}$ such that $p(S^*)=\max\{p(S):S\in\mathcal{F}\}$.
\end{definition}
We will consider that the empty set is always feasible, ensuring that the feasible set of solutions is non empty.
This is a very general class of problems, including the maximization \emph{Subset Selection} problems studied by Pruhs and Woeginger in \cite{Pruhs} (they only considered linear objective functions). It contains for instance graph problems where $N$ is the set of vertices (as in any maximization induced subgraph problem verifying some property) or the set of edges (as in  matching problems). It also contains classical set problems (knapsack, maximum 3-dimensional matching,\dots), and more generally 0-1 linear programs (with non negative profits in the objective function).





Given a problem in the previous class, we are interested in its \ot{} version~\cite{Gupta,Eisenstat}.
The stability over time of a solution sequence is classically captured by considering a transition cost when a modification is made in the solution. Here, dealing with maximization problems, we will consider a transition {\it bonus} $B$ for taking into account the similarity of two consecutive solutions.
In what follows, we will use the term object to denote an element of $N$ (so an object can be a vertex of a graph, or an edge,\dots, depending on the underlying problem). 	

\begin{definition}
\emph{(Multistage Subset Maximization Problems.)} In a Multistage Subset Maximization problem $\cal P$, we are given
\begin{itemize}
\item a number of steps $T \in \mathbb{N}$, a set $N$ of $n$ objects;
\item for any $t \in T$, an instance $I_t$ of the optimization problem. We will denote:
\begin{itemize}
\item $p_{t}$ the objective (profit) function at time $t$
\item $\mathcal{F}_t\in 2^N$ the set of feasible solutions at time $t$
\end{itemize} 
\item $B \in \mathbb{R^{+}}$ a given transition profit. 
\item the value of a solution sequence $\mathcal{S}=(S_1,\dots,S_T)$ is $$f(\mathcal{S})=\sum_{t=1}^T p_t(S_t) + \sum_{t=1}^{T-1} b(S_t,S_{t+1})$$
where $b(S_t,S_{t+1})$ is the transition bonus for the solution between time steps $t$ and $t+1$. We will use the term {\it profit} for $ p_t(S_t)$, {\it bonus} for the transition bonus $b(S_t,S_{t+1})$, and {\it value} of a solution $\mathcal{S}$ for $f(\mathcal{S})$;
\item the goal is to determine a solution sequence of maximum value. 
\end{itemize}
\end{definition}

 There are two natural ways to define the transition bonus. We will see that these two ways of measuring the stability induce some differences in the competitive ratios one can get.

\begin{definition}
\emph{(Types of transition bonus.)}
 If $S_t$ and $S_{t+1}$ denote, respectively, the solutions for time steps $t$ and $t+1$, then we can define the transition bonus as:
\begin{itemize}
    \item \emph{Intersection Bonus:}  $B$ times $|S_{t}\cap S_{t+1}|$: in this case the bonus is proportional to the number of objects in the solution at time $t$ that remain in it at time $t+1$.
    \item \emph{Hamming Bonus:} $B$ times $|S_{t}\cap S_{t+1}|+|\overline{S_{t}}\cap \overline{S_{t+1}}|$. Here we get the bonus for each object for which the decision (to be in the solution or not) is the same between time steps $t$ and $t+1$. In other words, the bonus is proportional to $|N|$ minus the number of modifications (Hamming distance) in the solutions.
\end{itemize}
\end{definition}

Note that by scaling profits (dividing them by $B$), we can arbitrarily fix $B=1$. So from now on, we assume $B=1$.



In this article, we will consider two possible ways for the data to evolve. 

\begin{definition}
\emph{(Types of data evolution.)}
\begin{itemize}
\item \emph{Static Set of Feasible Solutions (SSFS):}  only profits may change over time, so the structure of feasible solutions remains the same: $\mathcal{F}_t=\mathcal{F}$ for all $t$. 

\item \emph{General Evolution (GE):}  any modification in the input sequence is possible. Both the profits and the set of feasible solutions may change over time. In this latter model, for knapsack, profits and weights of object (and the capacity of the bag) may change over time; for maximum independent set edges in the graph may change,\dots.   
\end{itemize}
\end{definition}








\subsection{Related Work}

A series of papers consider the  online or semi-online settings, where the input changes over
time and the algorithm has to modify (re-optimize) the solution by making as few changes
as possible (see \cite{Anthony, Blanchard, Cohen, Gu, Megow, Nagarajan} and the references therein). 
The multistage model considered in this paper has been introduced in Eisenstat
et al. \cite{Eisenstat} and Gupta et al. \cite{Gupta}. Eisenstat
et al. \cite{Eisenstat} studied the multistage version of facility location problems. They proposed a logarithmic approximation algorithm. An et al. \cite{An} obtained constant factor approximation algorithms for some related problems.
 Gupta et al. \cite{Gupta} studied the {\sc Multistage  Maintenance Matroid} problem for
both the offline and the online settings. They presented a logarithmic approximation
algorithm for this problem, which includes as a special case a natural multistage version of
{\sc Spanning Tree}. They also considered the online version of the problem and they provide an efficient randomized competitive algorithm against any oblivious adversary. The same paper also introduced the study of the {\sc Multistage  Minimum Perfect Matching} problem for which they proved that it is hard to approximate 
 even for a constant number of stages.  Bampis et al. \cite{Bampis} improved this negative result by showing that the problem is
hard to approximate even for bipartite graphs and for the case of two time steps. When the edge costs are metric within every time step they  proved that the problem remains APX-hard even for two time steps. They also showed that the maximization version of the problem admits a constant factor approximation algorithm, but is APX-hard. Olver et al.~\cite{Olver} studied a multistage version of the {\sc Minimum Linear Arrangement} problem, which is related to a variant of the {\sc List Update} problem~\cite{Sleator}, and provided a logarithmic lower bound for the online version and a polylogarithmic upper bound for the offline version.

 The {\sc Multistage Max-Min Fair Allocation}
problem has been studied in the offline and  the online settings in \cite{Bampis+}. This problem corresponds to a multistage variant of the {\sc Santa Klaus} problem. For the off-line setting, the authors
showed that the multistage version of the problem is much harder
than the static one. They provided
constant factor approximation algorithms for the off-line setting. For the online setting they proposed a constant competitive ratio for SSFS-type evolving instances and they proved that it is not possible to find an online algorithm with bounded competitive ratio for GE-type evolving instances. Finally, they showed that in the 1-lookahead case, where at time step $t$ we know the instance of time step $t+1$, it is possible to get a constant approximation ratio.  

Buchbinder et
al. \cite{Buchbinder} and Buchbinder, Chen and Naor \cite{Buchbinder+} considered a multistage model and they studied the relation between the online learning and competitive analysis frameworks, mostly for fractional optimization problems. 

\subsection{Summary of Results and Overview}

The contribution of our paper is a framework for online multistage maximization problems (comprising different models), a characterization of those models in which a constant competitive ratio is achievable, and almost tight upper and lower bounds on the best-possible competitive ratio for these models.

We increase the complexity of the considered models over the course of the paper. We start with the arguably simplest model: Considering a static set of feasible solutions clearly restricts the general model of evolution; while such a straightforward comparison between the Hamming and intersection bonus is not possible, the Hamming bonus seems simpler in that, compared to the intersection model, there are (somewhat comparable) extra terms added on the profit of both the algorithm and the optimum. As we show in Subsection~\ref{subsec:static-hamming}, there is indeed a simple $2$-competitive algorithm: At each time $t$, it greedily chooses the set $S_t$ that either maximizes the transition bonus w.r.t.\ $S_{t-1}$ (that is, choosing $S_t=S_{t-1}$, which is possible in this model) or maximizes the value $p_t(S_t)$. We complement this observation with a matching lower bound only involving two time steps.

We then toggle the transition-bonus model and the data-evolution model separately and show that constant competitive ratios can still be achieved. First, in Subsection~\ref{subsec:static-intersection}, we consider intersection bonus. We show that, \emph{after} modifying the profits to make larger solutions more profitable, a $(2+1/(T-1))$-competitive algorithm can be achieved by a greedy approach again. We also give an (almost matching) lower bound of $2$ again. Next, we toggle the evolution model. In Subsection~\ref{subsec:general-hamming}, we adapt the greedy algorithm from Subsection~\ref{subsec:static-hamming} by reweighting to obtain a $(3+1/(T-1))$-competitive algorithm using a more complicated analysis. We complement this result with a lower bound of $1+\sqrt{2}$.

In Subsection~\ref{subsec:general-intersection}, we finally consider the general-evolution model with intersection bonus, where we give a simple lower bound showing that a constant-competitive ratio is not achievable. This lower bound relies on forbidding to choose any item in the second step that the algorithm chose in the first step. We circumnavigate such issues by allowing the algorithm a lookahead of one step and present a $4$-competitive algorithm for that setting. A similar phase transition has been observed for a related problem~\cite{Bampis+}, but our algorithm, based on a doubling approach, is different. We also give a matching lower bound of $4$ on the competitive ratio of any algorithm in the same setting. We summarize all results described thus far in Table~\ref{table:results}.


\begin{table}[t]
	\centering
	\def\arraystretch{1.1}
	\begin{tabular}{ccc}
		\toprule
		&static set of feasible solutions&general evolution\\
		\midrule
		\multirow{2}{*}{Hamming bonus}& $c^\star=2$ & $1+\sqrt{2}\leq c^\star\leq3+o(1)$ \\
		&  Theorems \ref{thm:static-hamming-lower} and \ref{thm:static-hamming-upper} & Theorems \ref{thm:general-hamming-lower} and \ref{thm:general-hamming-upper}\\[1ex]  
		\midrule
		\multirow{3}{*}{intersection bonus}& \multirow{2}{*}{$2\leq c^\star\leq 2+o(1)$} & $c^\star=\infty$\\
		& & $c^\star=4$ for $1$-lookahead\\		
		&  Theorems \ref{thm:static-intersection-lower} and \ref{thm:static-intersection-upper} & Theorems \ref{thm:general-intersection-lower}, \ref{thm:general-intersection-lower1}, and \ref{thm:general-intersection-upper} \\[1ex]  
		\bottomrule
	\end{tabular}
	\medskip
	\caption{Our bounds on the best-possible competitive ratio $c^\star$ for the different models. The Landau symbol is with respect to $T\rightarrow\infty$.}
	\label{table:results}
\end{table}

We note that the lower bounds mentioned for the Hamming model are only shown for a specific fixed number of time steps, and that in general there is no trivial way of extending these bounds to a larger number of time steps. 
One may however argue that the large-$T$ regime is in fact the interesting one for both practical applications and in theory, the latter because the effect of having a first time step without bonus vanishes. At the end of the respective sections, we therefore give asymptotical lower bounds of $3/2$ and roughly $1.696$ for the cases of a static set of feasible solutions and general evolutions, respectively. These bounds are non-trivial, but we do not know if they are tight.

It is plausible that the aforementioned upper bounds can be improved if extra assumptions on characteristics of the objective function and the sets of feasible solutions are made. In Subsubsection~\ref{subsec:submodular}, we show that already very natural assumptions suffice: Assuming that at each time the feasible solutions are closed under taking subsets and the objective function is submodular, we give a $(21/8+o(1))$-competitive algorithm for the model with a general evolution and Hamming bonus, improving the previous $(3+o(1))$-competitive ratio. Our  lower bounds for general evolution and Hamming bonus in fact fulfill the extra assumptions.

In Section~\ref{sec:conclusion}, we summarize our results and mention directions for future research that we consider interesting.




\section{Model of a Static Set of Feasible Solutions}\label{sec:static}

We consider here the model of evolution where only profits change over time: $\mathcal{F}_t=\mathcal{F}$ for any $t$. 
We first consider the \easy{} bonus model and show a simple 2-competitive algorithm. We will then show that a (asymptotic) competitive ratio of 2 can also be achieved in the intersection bonus model using a more involved algorithm. In both cases, this ratio 2 is shown to be (asymptotically) optimal.  

\subsection{\easy{}-Bonus Model}
\label{subsec:static-hamming}

\begin{theorem}\label{thm:static-hamming-upper}
In the SSFS model with Hamming bonus, there is a 2-competitive algorithm. 
\end{theorem}

\begin{proof}
We consider the very simple following algorithm. At each time step $t$, the algorithm computes an optimal solution $S^*_t$ with associated profit $p_t(S^*_t)$. At $t=1$ we fix $S_1=S^*_1$. For $t>1$, if $p_t(S^*_t) > n$ then fix $S_t=S^*_t$, otherwise fix $S_t=S^*_{t-1}$ (which is possible thanks to the fact that the set of feasible solutions does not change).

Let $f^*$ be the optimal value. Since any solution sequence gets profit at most $p_t(S^*_t)$ at time $t$, and bonus at most $n$ between two consecutive time steps, we get $f^* \leq \sum_{t=1}^T p(S^*_t) + n(T-1)$. 

By construction, at time $t>1$, either the algorithm gets profit $p_t(S^*_t)$ when $p_t(S^*_t)>n$, or bonus (from $t-1$) $n$ when $n\geq p_t(S^*_t)$. So in any case the algorithm gets profit plus bonus at least $\frac{p_t(S^*_t)+n}{2}$. At time $1$ it gets profit at least $p_1(S^*_1)$. So 

$$f(S_1\dots,S_T)\geq p_1(S^*_1)+\sum_{t=2}^T \frac{p_t(S^*_t)}{2}+ \frac{n(T-1)}{2}\geq \frac{f^*}{2},$$
which completes the proof.\end{proof}

\begin{theorem}\label{thm:static-hamming-lower}
Consider the SSFS model with Hamming bonus. For any $\epsilon>0$, there is no $(2-\epsilon)$-competitive algorithm, even if there are only $2$ time steps. 
\end{theorem}

\begin{proof}
We consider a set $N=\{1,2,\dots,n\}$ of $n=1+\left\lceil \frac{1}{\epsilon}\right\rceil$ objects, and $T=2$ time steps. There are three feasible solutions: $S^0=\emptyset$, $S^1=\{1\}$ and $S^2=\{2,\dots,n\}$. At $t=1$, all the profits are 0. Let us consider an on-line algorithm \texttt{A}. We consider the three possibilities for the algorithm at time 1:
\begin{itemize}
\item At time 1, \texttt{A} chooses $S^0$: at time 2 we give profit 1 to all objects. If \texttt{A} takes no object at time 2, it gets profit 0 and bonus $n$. If it takes $S^1$, it gets profit 1 and bonus $n-1$. If it takes $S^2$, it gets profit $n-1$ and bonus 1, so in any case the computed solution has value $n$. The solution consisting of taking $S^2$ at both time steps has profit $n-1$ and bonus $n$, so value $2n-1$.
\item At time 1, \texttt{A} chooses $S^1$: at time 2 we give profit 0 to object 1, and profit 1 to all other objects. Then, if the algorithm takes $S^0$ (resp, $S^1$, $S^2$), at time 2 its gets value $n-1$ (resp, $n$, $n-1$) while the solution consisting of taking $S^2$ at both time steps has value $2n-1$.
\item At time 1, \texttt{A} chooses $S^1$: at time 2 we give profit $n$ to object 1, and 0 to all other objects. Then if the algorithm takes $S^0$ (resp, $S^1$, $S^2$) at time 2 its gets value $n-1$ (resp, $n$, $n$), while  the solution consisting of taking $S^1$ at both time steps has value $2n$.
\end{itemize}
In any case, the ratio is at least $\frac{2n-1}{n}=2-\frac{1}{n}>2-\epsilon$.
\end{proof}

We complement this lower bound with an asymptotical result for large $T$.

\begin{theorem}\label{th:lbmod1}
	Consider the SSFS model with Hamming bonus. For every $\epsilon>0$, there is a $T_\epsilon$ such that, for each number of time steps $T\geq T_\epsilon$, there is no $(3/2-\epsilon)$-competitive algorithm.
\end{theorem}
\begin{proof}
	Let $N:=\{1,2\}$. The static set of feasible solutions is $\mathcal{F}=\{\emptyset,\{1\},\{2\}\}$. Initially, $p_1(\{1\})=0$ and $p_1(\{2\})=1$. As long as the algorithm has not picked item $1$ until some time $t$, we set $p_{t+1}(\{1\})=0$ and $p_{t+1}(\{2\})=1$ again. Note that, in order to be $(3/2-\epsilon)$-competitive, the algorithm however has to pick item $1$ eventually. Further, the ratio between the profit of the optimum and the algorithm during this part is $3/2-o(1)$ as the length of this part approaches $\infty$.
	
	The remaining time horizon is partitioned into contiguous \emph{phases}. Consider a phase that starts at time $t$. The invariant at the beginning of the phase is that both the algorithm and the optimum have picked the same item in the previous time step $t-1$. Let this item be w.l.o.g.\ item $2$; the other case is symmetric. Then $p_t(\{2\})=1$ and $p_t(\{1\})=3$. By the same reasoning as above, we can assume the algorithm chooses an item at $t$. Let $i\in\{1,2\}$ be that item. Then $p_{t+1}(\{i\})=0$ and $p_{t+1}(\{3-i\})=1$. As long as the algorithm is still not picking item $3-i$ during the time interval $[t+1,t']$, $p_{t'+1}(\{i\})=0$ and $p_{t'+1}(\{3-i\})=1$. Once the algorithm picks item $3-i$ at some time, the phase ends \emph{regularly}; otherwise it ends \emph{by default}.
	
	Now consider a phase of length $\ell$ that ends regularly (note $\ell\geq 2$). We claim that the values of the algorithm and the optimum have a ratio of at least $3/2$. This is because of the following estimates on the algorithm's and optimum's value:
	\begin{itemize}
	    \item In either case for $i$, the algorithm obtains a value of $3$ in time step $t$. Furthermore, the total bonus in all subsequent time steps is $(\ell-2)\cdot 2$, because the algorithm has to switch from item $i$ to item $3-i$. There is an additional profit of $1$ at time $t+\ell-1$. Therefore, the total value is $4+(\ell-2)\cdot 2$
	    \item The value of the optimum is at least $6+(\ell-2)\cdot 3$: It chooses item $3-i$ already at time $t$ and keeps it until time $t+\ell-1$, obtaining a value of $3$ in that time step and another $3$ in each subsequent time step.
	\end{itemize}
	This proves the claim and thereby the theorem as a phase that ends by default can be extended to one that ends regularly by modifying the optimum's and algorithm's values by constants.
\end{proof}

\subsection{\hard{}-Bonus Model}
\label{subsec:static-intersection}

In the intersection-bonus model things get harder since an optimal solution $S^*_t$ may be of small size and then gives very small (potential) bonus for the next step. As a matter of fact, the algorithm of the previous section has unbounded competitive ratio in this case: take a large number $n$ of objects, $\mathcal{F}=2^N$, and at time 1 all objects have profit 0 up to one which has profit $\epsilon$. The algorithm will take this object (instead of taking $n-1$ objects of profit 0) and then potentially get bonus at most 1 instead of $n-1$.

Thus we shall put an incentive for the algorithm to take solutions of large size, in order to have a chance to get a large bonus. We define the following algorithm called \texttt{MP-Algo} (for Modified Profit algorithm). Informally, at each time step $t$, the algorithm computes an optimal solution with a modified objective function $p'_{t}$. These modifications take into account (1) the objects taken at time $t-1$ (2) an incentive to take a lot of objects. Formally, \texttt{MP-Algo} works as follows:

\begin{enumerate}
    \item At $t=1$: let $p'_{1}(S)=p_{1}(S)+|S|$. Choose $S_1$ as an optimal solution for the problem with modified profits $p'_{1}$. 
    \item For $t$ from 2 to $T-1$: let $p'_{t}(S)=p_{t}(S)+|S\cap S_{t-1}|+|S|$. Choose $S_t$ as an optimal solution for the problem  with modified profit function $p'_{t}$. 
    \item At $t=T$: let $p'_{T}(S)=p_{T}(S)+|S\cap S_{T-1}|$. Choose $S_T$ as an optimal solution with modified profit function $p'_{T}$. 
\end{enumerate}
The cases $t=1$ and $t=T$ are specific since there is no previous solution for $t=1$, and no future solution for $t=T$.

\begin{theorem}\label{thm:static-intersection-upper}
In the SSFS model with intersection bonus, \texttt{MP-Algo} is $\left(\frac{2}{1-1/(T-1)}\right)$-competitive.
\end{theorem}

\begin{proof}

 Let $(\hat{S}_1,\dots,\hat{S}_T)$ be an optimal sequence. Since $S_t$ is optimal with respect to $p'_t$, for $t=2,\dots,T-1$ we have:
 \begin{equation}\label{eqBruno1}
 p'_t(S_t)=p_t(S_t)+|S_t\cap S_{t-1}|+|S_t|\geq p'_t(\hat{S}_t)\geq p_t(\hat{S}_t)+|\hat{S}_t|.
 \end{equation}
Since $S_{t-1}$ is also a feasible solution at time $t$, we have:
 \begin{equation}\label{eqBruno2}
 p'_t(S_t)=p_t(S_t)+|S_t\cap S_{t-1}|+|S_t|\geq p_t(S_{t-1}) \geq 2|S_{t-1}|.
 \end{equation}
Similarly, at $t=T$  $p'_T(S)=p_T(S)+|S\cap S_{t-1}|$ so 
 \begin{eqnarray}
 p_T(S_T)+|S_T\cap S_{T-1}| & \geq & p_T(\hat{S}_T), \label{eqBruno3}\\
  p_T(S_T)+|S_T\cap S_{T-1}|&\geq& |S_{T-1}|. \label{eqBruno4}
 \end{eqnarray}
At $t=1$  $p'_t(S)=p_t(S)+|S|$, so 
 \begin{equation}\label{eqBruno5}
 p_1(S_1)+|S_1|\geq p_1(\hat{S}_1)+|\hat{S}_1|.
 \end{equation}
Now, note that $|S_t\cap S_{t-1}|$ is the transition bonus of the computed solution between $t-1$ and~$t$. By summing Equation~(\ref{eqBruno1}) for $t=2,\dots,T-1$, Equation~(\ref{eqBruno3}) and Equation~(\ref{eqBruno5}), we deduce:
\begin{equation}\nonumber
    f(S_1,\dots,S_T)+\sum_{t=1}^{T-1}|S_t| \geq \sum_{t=1}^T p_t(\hat{S}_t)+\sum_{t=1}^{T-1}|\hat{S}_t|.
\end{equation}
Since in the optimal sequence the transition bonus between time $t$ and $t+1$ is at most $|\hat{S}_t|$, we get:
\begin{equation}\label{eqBruno6}
    f(S_1,\dots,S_T)+\sum_{t=1}^{T-1}|S_t| \geq f(\hat{S}_1,\dots,\hat{S}_T).
\end{equation}
Now we sum Equation~(\ref{eqBruno2}) for $t=2,\dots,T-1$ and Equation~(\ref{eqBruno4}):
\begin{equation}\nonumber
    f(S_1,\dots,S_T)+\sum_{t=2}^{T-1}|S_t| \geq 2\sum_{t=2}^{T-1}|S_{t-1}|+|S_{T-1}|.
\end{equation}
From this we easily derive:
\begin{equation}\label{eqBruno7}
    f(S_1,\dots,S_T) \geq \sum_{t=2}^{T-2}|S_{t}|.
\end{equation}
By summing Equations~(\ref{eqBruno6}) and~(\ref{eqBruno7}) we have $2f(S_1,\dots,S_T)\geq f(\hat{S}_1,\dots,\hat{S}_T)-|S_{T-1}|$. The competitive ratio follows from the fact that $f(\hat{S}_1,\dots,\hat{S}_T)\geq (T-1)|S_{T-1}|$ (since $S_{T-1}$ is feasible for all time steps).
\end{proof}

We note that competitive ratio 2 can be derived with a similar analysis when 
the number of time steps is 2 or 3. We show a matching lower bound (which is also valid in the asymptotic setting).

\begin{theorem}\label{thm:static-intersection-lower}
Consider the SSFS model with intersection bonus. For any $\epsilon>0$ and number of time steps $T=\lceil1/\epsilon\rceil$, there is no $(2-\epsilon)$-competitive algorithm.
\end{theorem}
\begin{proof}
Let $\epsilon>0$ and $T=\left\lceil \frac{1}{\epsilon}\right\rceil$. We consider $T$ time steps,  and a set $N$ of $n=T$ objects. The objective function is linear, and feasible solutions are sets of at most 1 object. At $t=1$, the profit of each object is 1. Then, at each time step, if the algorithm takes an object, this object will have profit 0 until the end. While an object is not taken by the algorithm, its profit remains 1.

Since the algorithm takes at most one object at each time step, there is an object which is never taken till the last step. The solution of taking this object during all the process has value $2T-1$. But at each time step the algorithm either takes a new object (and gets no bonus) or keeps the previously taken object and gets no profit. So the value of the computed solution is at most $T$. The ratio is $2-\frac{1}{T}\geq 2-\epsilon$.
\end{proof}


\section{Model of General Evolution}\label{sec:general}

We consider in this section that the  set of feasible solutions may evolve over time. We will show that in the Hamming bonus model, we can still get constant competitive ratios, though ratios slightly worse than in the case where only profits could change over time. Then, we will tackle the intersection bonus model, showing that no constant competitive ratio can be achieved. However, with only $1$-lookahead we can get a constant competitive ratio.

\subsection{\easy{}-Bonus Model}
\label{subsec:general-hamming}

In this section we consider the \easy{} bonus model. We first show in Section~\ref{subsec:general} that there exists a $\left(3+\frac{1}{T-1}\right)$-competitive algorithm.
Interestingly, we then show in Section~\ref{subsec:submodular} that a slight assumption on the problem structure allows to improve the competitive ratio. More precisely, we achieve a 21/8 (asymptotic) competitive ratio if we assume that the objective function is submodular (including the additive case) and that a subset of a feasible solution is feasible. These assumptions are satisfied by all the problems mentioned in introduction.
We finally consider lower bounds in Section~\ref{sub:generallower}.

\subsubsection{General Case}\label{subsec:general}

We adapt the idea of the 2-competitive algorithm working for the Hamming bonus model for a static set of feasible solutions (Section~\ref{subsec:static-hamming}) to the current setting where the set of feasible solutions may change.  Let us consider the following algorithm \texttt{BestOrNothing}: at each time step $t$, \texttt{BestOrNothing} computes an optimal solution $S_t^*$ with associated profit $p_t(S_t^*)$ and compares it to $2$ times the maximum potential bonus, i.e to $2n$. It chooses $S^*_t$ if the associated profit is at least $2n$, otherwise it chooses $S_t=\emptyset$. A slight modification is applied for the last step  $T$.

Formally, \texttt{BestOrNothing}  works as follows: 

\begin{enumerate}
\item For $t$ from 1 to $T-1$:
\begin{enumerate}
	\item Compute an optimal solution $S^*_t$ at time $t$ with associated profit $p_t(S_t^*)$ 
	\item If $p_t(S_t^*)\geq 2n$ set $S_t=S^*_t$, otherwise set $S_t=\emptyset$.
\end{enumerate}
\item At time $T$:
	\begin{enumerate}
	\item if $S_{T-1}=S^*_{T-1}$, then $S_T=S^*_T$. 
	\item Otherwise: if $p_t(S_t^*)\geq n$ set $S_T=S^*_T$, otherwise set $S_T=\emptyset$.
	\end{enumerate}
\end{enumerate}

We shown an upper bound on the competitive ratio achieved by this algorithm.

\begin{theorem}\label{thm:general-hamming-upper}
In the GE model with Hamming bonus, \texttt{BestOrNothing} is $\left(3+\frac{1}{T-1}\right)$-competitive.
\end{theorem}

\begin{proof}[Proof of Theorem~\ref{thm:general-hamming-upper}]
Let us define $J\subseteq\{1,\dots,T-1\}$ as the set of time steps $t<T$ where $p_t(S_t^*)\geq 2n$. 

If $J\neq \emptyset$, let $t_1$ be the largest element in $J$. We first upper bound the loss of the algorithm up to time $t_1$. We will then deal with the time period from $t+1$ up to $T$.

The global profit of an optimal solution up to time $t_1$ is at most $2n(t_1-|J|)+\sum_{t\in J}p_t(S_t^*)$. Its bonus (including the one from time $t_1$ to $t_1+1$) is at most $nt_1$. So its global value is at most $n\left(3t_1-2|J|\right)+\sum_{t\in J}p_t(S_t^*)$.

The solution computed by \texttt{BestOrNothing} gets profit at least $\sum_{t\in J}p_t(S_t^*)$. Note that it chooses the empty set always but $|J|$ times, so it gets transition bonus $n$ at least $t_1-2|J|$ times (each step in $J$ may prevent to get the bonus only between $t-1$ and $t$, and between $t$ and $t+1$). So the global value of the computed solution up to time $t_1$ is at least $n\max\{0;t_1-2|J|\}+\sum_{t\in J}p_t(S_t^*)$.

Up to time $t_1$, the ratio $r$ between the optimal value and the value of the solution computed by \texttt{BestOrNothing} verifies $$r\leq \frac{n\left(3t_1-2|J|\right)+\sum_{t\in J}p_t(S_t^*)}{n\max\{0;t_1-2|J|\}+\sum_{t\in J}p_t(S_t^*)}\leq \frac{3t_1}{\max\{0;t_1-2|J|\}+2|J|},$$ 
where we used the fact that $\sum_{t\in J}p_t(S_t^*)\geq 2n|J|$.
Since $\max\{0;t_1-2|J|\}+2|J|\geq t_1$ the ratio is at most 3 up to time $t_1$.\\

Now, let us consider the end of the process, from time $t_1+1$ (or 1 if $J$ is empty) up to time $T$. If $t_1=T-1$ then we take the best solution at time $T$ and get no extra loss, so the algorithm is 3-competitive in this case. 

Now assume $t_1<T-1$. We know that \texttt{BestOrNothing} chooses the empty set up to $T-1$. Let us first assume that $p_T(S^*_T)<n$. Then on the subperiod from $t_1+1$ to $T$ \texttt{BestOrNothing} gets value $n(T-t_1-1)$ (bonuses), while the optimum gets bonus at most $n(T-t_1-1)$ and profit at most $2n(T-t_1-1)+n$. The optimal value is then at most $n\left( 3T-3t_1-2\right)\leq 3n(T-t_1-1)+n$.

Now suppose that $p_T(S^*_T)\geq n$. On the subperiod from $t_1+1$ to $T$ \texttt{BestOrNothing} gets value $n(T-t_1-2)+p_T(S^*_T)$, while the optimum gets bonus at most $n(T-t_1-1)$ and profit at most $2n(T-t_1-2)+p_T(S^*_T)$. The worst case ratio occurs when $p_T(S^*_T)= n$. In this case, as before, the value of the computed solution is $n(T-t_1-1)$, while the optimal value is at most $n\left( 3T-3t_1-2\right)\leq 3n(T-t_1-1)+n$.

Then, in all cases we have that the optimal value is at most $3f(S_1,\dots,S_n)+n$. But $f(S_1,\dots,S_n)\geq (T-1)n$ (the computed solution has value at least $t_1n$ up to $t_1$, and then at least $n(T-t_1-1)$), and the claimed ratio follows.
\end{proof}

\subsubsection{Improvement for Submodularity and Subset Feasibility}\label{subsec:submodular}

In this section we assume that the problem have the following two properties:
\begin{itemize}
    \item {\it subset feasibility}: at any time step, every subset of a feasible solution is feasible.
    \item {\it submodularity}: for any $S,S'$, any $t$, $p_t(S\cap S')+p_t(S\cup S')\leq p_t(S)+p_t(S')$.
\end{itemize}
Note that this implies that, if a feasible set $X$ is partitioned into (disjoint) subsets $X_1,\dots,X_h$, then $X_1,\dots,X_h$ are feasible and $p_t(X)\leq \sum_i p_t(X_i)$.

We exploit this property to devise algorithms where we partition the set of objects and solve the problems on subinstances. As a first idea, let us partition the set of objects into into $3$ sets $A,B,C$ of size (roughly) $n/3$;  consider the algorithm which at every time step $t$ computes the best solutions $S^A_t,S^B_t,S^C_t$ on each subinstance on $A$, $B$ and $C$, and chooses $S_t$ as the one of maximum profit between these 3 solutions. By submodularity and subset feasibility, the algorithm gets profit at least 1/3 of the optimal profit at each time step. Dealing with bonuses, at each time step the algorithm chooses a solution included either in $A$, or in $B$, or in $C$ so, for any $t<T$, at least one set among $A, B \text{ and } C$ is not chosen neither at time $t$ nor at time $t+1$, and the algorithm gets transition bonus at least $n/3$. Hence, the algorithm is 3-competitive.

We now improve the previous algorithm. The basic idea is to remark that if for two consecutive time steps $t,t+1$ the solution $S_t$ and $S_{t+1}$ are taken in the same subset, say $A$, then the bonus is (at least) $2n/3$ instead of $n/3$. Roughly speaking, we can hope for a ratio better than 1/3 for the bonus. Then the algorithm makes a tradeoff at every time step: if the profit is very high then it will take a solution maximizing the profit, otherwise it will do (nearly) the same as previously.   
More formally, let us consider the algorithm \texttt{3-Part}. We first assume that $n$ is a multiple of 3. $x\in [0,1]$ will be defined later.

\begin{enumerate}
    \item Partition $N$ in three subsets $A,B,C$ of size $n/3$.
	\item For $t \in \{1, \ldots,T\}$: compute a solution $S^*_t$ maximizing $p_t(S)$
	\begin{itemize}
		\item Case (1): If $p_t(S^*_t) \geq xn$: define $S_t=S^*_t$
		\item Otherwise ($p_t(S^*_t) \leq xn$): compute solutions with optimal profit $S^A_t$, $S^B_t$, $S^C_t$ included in $A$, $B$ and $C$. Let $a_t$, $b_t$ and $c_t$ the respective profits.
		\begin{itemize}
			\item Case (2): if $t\geq 2$ and Case (1) did not occur at $t-1$, do: 
			
			If $S_{t-1}\subseteq A$ (resp. $S_{t-1}\subseteq B$, $S_{t-1}\subseteq C$), compute $\max\{a_t+2n/3, b_t+n/3,c_t+n/3\}$ (resp. $\max\{a_t+n/3, b_t+2n/3,c_t+n/3\}$, $\max\{a_t+n/3, b_t+n/3,c_t+2n/3\}$) and define $S_t$ as $S^A_t$, $S^B_t$ or $S^C_t$ accordingly. 
			\item Case (3) ($t=1$ or Case (1) occurred at $t-1$) do:
			\begin{itemize}
				\item Define $S_t$ as the solution with maximum profit among $S^A_t$, $S^B_t$, $S^C_t$.
			\end{itemize}	
		\end{itemize}
	\end{itemize}
\end{enumerate}

If $N$ is not a multiple of 3, we add one or two dummy objects that are in no feasible solutions (at any step). We prove an upper bound on the competitive ratio of this algorithm.

\begin{theorem}\label{theo:sumod2}
Consider the GE model with Hamming bonus. Under the assumption of subset feasibility and submodularity, \texttt{3-Part} is $(21/8+O(1/T+1/n))$-competitive.
\end{theorem}
\begin{proof}
We mainly show that in each case (1), (2) or (3) the computed solution achieves the claimed ratio.
\begin{itemize}
    \item  Let us first consider a time step $t\geq 2$ where Case (2) occurs. It means that Case (2) or (3) occurred at the previous step, so $S_{t-1}$ is included in $A$, $B$ or $C$. Suppose w.l.o.g that algorithm took $S_{t-1}\subseteq A$. Then $S^A_t$ gives a bonus at least $2n/3$ (between $t-1$ and $t$), and $S^B_t$ and  $S^C_t$ gives a bonus at least $n/3$. By computing $\max\{a_t+2n/3, b_t+n/3,c_t+n/3\}$, we derive:
\begin{equation}\nonumber
p_t(S_t)+b(S_t,S_{t-1})\geq \frac{a_t+2n/3+ b_t+n/3+c_t+n/3}{3}\geq \frac{p_t(S^*_t)}{3}+\frac{4n}{9},
\end{equation}
where $S^*_t$ is a solution maximizing the profit at time $t$, using the fact that $p_t(S^*_t)\leq a_t+b_t+c_t$ by subset feasibility and submodularity.
Since in Case (2) $p_t(S^*_t)\leq xn$, we derive:
\begin{equation}\label{eqtcas2}
p_t(S_t)+b(S_t,S_{t-1})\geq  r_2 \left(p_t(S^*_t)+n\right)
\end{equation}
with $r_2=\frac{3x+4}{9(1+x)}$.
\item Now, consider a time step $t\geq 2$ where Case (3) occurs. Then necessarily Case (1) occurs at step $t-1$. So $S_{t-1}=S^*_{t-1}$. Also, $S_t$ has profit at least $p_{t}(S^*_{t})/3$. So 
\begin{equation}\nonumber
    \sum_{\ell=t-1}^t p_\ell(S_\ell)+b(S_\ell,S_{\ell-1})\geq  p_{t-1}(S^*_{t-1}) + \frac{p_t(S^*_t)}{3}
\end{equation}
So $\sum_{\ell=t-1}^t p_\ell(S_\ell)+b(S_\ell,S_{\ell-1})\geq  r_3\left (p_{t-1}(S^*_{t-1})+n+ p_t(S^*_t)+n\right)$
with $$r_3=\frac{p_{t-1}(S^*_{t-1}) + \frac{p_t(S^*_t)}{3}}{p_{t-1}(S^*_{t-1})+2n+ p_t(S^*_t)}.$$ Since $p_{t-1}(S^*_{t-1})\geq xn$, we get:
\begin{equation}\nonumber
    r_3 \geq \frac{xn + \frac{p_t(S^*_t)}{3}}{(2+x)n+ p_t(S^*_t)}.
\end{equation}
Since $p_{t}(S^*_{t})\leq xn$, provided that we choose $x\geq 1$ such that $x/(2+x)\geq 1/3$, we get:
\begin{equation}\nonumber
    r_3 \geq \frac{xn + xn/3}{(2+x)n+ xn}=\frac{2x}{3(1+x)}.
\end{equation}

\item Finally, suppose that Case (1) occurs at some step $t\geq 2$. Then $S_t=S^*_t$ and $p(S^*_t)\geq xn$, so 
\begin{equation}\nonumber
    p_t(S_t)+b(S_t,S_{t-1})\geq  p_t(S_t) = p_t(S^*_t)\geq r_1(p_t(S^*_t)+n).
\end{equation}
with $r_1=\frac{x}{1+x}$.
\end{itemize}
By setting $x=\frac{4}{3}$, we get $r_1\geq r_2=r_3=8/21$.

It remains to look at step 1. If $p_1(S^*_1)\geq xn$ (Case (1)), then $S_1=S^*_1$, so there is no profit loss. Otherwise, $p_1(S^*_1)\leq xn$, Case (3) occurs, the loss it at most $2p_1(S^*_1)/3\leq 2xn/3\leq n$. Since the optimal value is at least $n(T-1)$, the loss it a fraction at most $1/(T-1)$ of the optimal value.

If $n$ is not a multiple of 3, adding one or two dummy objects add $T-1$ or $2(T-1)$ to solution values, inducing a loss which is a fraction at most $O(1/n)$ of the optimal value.
\end{proof}

\subsubsection{Lower Bounds}\label{sub:generallower}

We complement the algorithmic results with a lower bound for two time steps and an asymptotical one. Interestingly, these bounds are also valid for the latter restricted setting with subset feasibility and submodularity.

\begin{theorem}\label{thm:general-hamming-lower}
Consider the GE model with Hamming bonus. For any $\epsilon>0$, there is no $(1+\sqrt{2}-\epsilon)$-competitive algorithm.
\end{theorem}
\begin{proof}
We consider a knapsack problem with $n$ objects ($n=2$ suffices to show the result, but the proof is valid for any number $n$ of objects), and $T=2$ time steps. At time 1, all objects have weight 1 and profit $\alpha=\sqrt{2}-1$; the capacity of the bag is $n$. 

Let $S_1$ be the set of objects chosen at step 1 by the algorithm (possibly $S_1=\emptyset$). At $t=2$ the algorithm receives the instance $I_2(S_1)$ where:
\begin{itemize}
	\item the capacity is $c_2=n-|S_1|$.
	\item each object not in $S_1$ receives weight and profit 1.
	\item each object in $S_1$ has a weight greater than $c_2$.
\end{itemize}
Then at step 2 the algorithm receives value 1 for each object not in $|S_1|$ (either by transition bonus from step 1, or by taking it at step 2). The value of its solution is $\alpha |S_1|+n-|S_1|$. Now, the solution consisting of taking $\overline{S_1}$ at both time steps has value $\alpha(n-|S_1|)+n-|S_1|+n=n(2+\alpha)-|S_1|(1+\alpha)$. The chosen $\alpha$ is such that $2+\alpha=\frac{1+\alpha}{1-\alpha}$, so the solution $(\overline{S},\overline{S})$ has value $\frac{1+\alpha}{1-\alpha}\left(n-|S|(1-\alpha)\right)$. The ratio is $\frac{1+\alpha}{1-\alpha}=2+\alpha=1+\sqrt{2}$.
\end{proof}

\begin{theorem}~\label{th:lwasymp}
	Consider the GE model with Hamming bonus. For every $\epsilon>0$, there is a $T_\epsilon$ such that, for each number of time steps $T\geq T_\epsilon$, there is no $(\alpha-\epsilon)$-competitive algorithm where $\alpha= \frac{6\cdot\sqrt[3]{9 + \sqrt{87}}}{\sqrt[3]{6\cdot(9 + \sqrt{87})^2}-\sqrt[3]{36}}\approx1.696$.
\end{theorem}
\begin{proof}
	Consider some $\epsilon>0$ and some online algorithm \texttt{A}. The ground set only consists of the single item $1$, that is, $N=\{1\}$.
	
	At time $1$, it is not feasible to pick the item, that is, $\mathcal{F}_1=\{\emptyset\}$. We partition the remaining time horizon $\{2,3,\dots,T\}$ (with $T$ yet to be specified) into \emph{phases}. Hence, the first phase starts in time step $2$. In any phase, as long as \texttt{A} has not included item $1$ in its solution until time $t<T$, both including and not including it is feasible at $t+1$, that is, $\mathcal{F}_{t+1}=\{\emptyset,\{1\}\}$. Once \texttt{A} includes the item in its solution at time $t<T$ (meaning $S_{t}=\{1\}$), including it becomes unfeasible at the next time, that is, $\mathcal{F}_{t+1}=\{\emptyset\}$. The current phase also ends at this time. In this case, we say that the phase ends \emph{regularly}. At $t=T$, the current phase ends \emph{by default} in any case. If a phase however ends regularly at time $t+1<T$, a new phase starts at time $t+2$.
	
 There is no profit associated with the empty set, that is, $p_t(\emptyset)=0$; the profit $p_t(\{1\})$ is $\beta$ whenever $t$ is the first time step of a phase, and it is $\gamma$ in all other cases (note that, however, it may be unfeasible to include item $1$ in the solution).  The remaining part of the proof is concerned with finding $\beta,\gamma$ so as to maximize the competitive ratio. 
	
	For the analysis, denote by $\mathcal{S}^*=(S^*_1,S^*_2,\dots,S^*_T)$ the optimal solution, and denote by $\mathcal{S}=(S_1,S_2,\dots,S_T)$ the solution that \texttt{A} finds. We consider phases separately. First consider a phase of length $\ell$ starting at time $t_0$ ending regularly (at time $t_0+\ell-1$). Note that $\ell\geq2$ and that the initial situation is independent of $t_0$ and $\ell$ in that $1\notin S_{t-1}$. For each time $t$ that is part of the phase, we count $b(S_{t-1},S_t)+p_t(S_t)$ and $b(S^*_{t-1},S^*_t)+p_t(S^*_t)$ towards the values of the optimum and algorithm, respectively. If $\ell=2$, the resulting values of the algorithm and optimum are $\max\{\beta,2\}\geq2$ and $\beta$, respectively. If $\ell>2$, the value are $\max\{\beta+(\ell-2)(1+\gamma),\ell\}\geq\beta+(\ell-2)(1+\gamma)$ and $(\ell-2)+\gamma$, respectively. Hence, in phases of length at most $2$, the optimum does not pick the item; in longer phases, it picks the item at all times when it can.
	
	To express the lower bound that we can show, first note that assuming that each phase ends regularly is only with an additive constant loss in both the algorithm's and the optimum's value, so we may make this assumption for the asymptotical competitive ratio considered here. Since the algorithm chooses the phase lengths, the lower bound $\alpha$ that we can show here is equal to the largest lower bound on the ratio between the optimum's and the algorithm's value within any phase, which is lower bounded by \begin{equation}\label{eq:LB-1}
	    \min\left\{\frac{2}{\beta},\inf_{\ell\in \mathbb{N};\ell\geq3}\frac{\beta+(\ell-2)(1+\gamma)}{(\ell-2)+\gamma}\right\}
	\end{equation}
	according to the above considerations.
	
	Note that the infimum in~\eqref{eq:LB-1} is minimized when its argument is identical across all $\gamma$. This is the case when $$\frac{1+\beta+\gamma}{1+\gamma}=1+\gamma\Leftrightarrow \gamma=\frac12\cdot(\sqrt{4\beta+1}-1).$$ Furthermore,~\eqref{eq:LB-1} is minimized when both its arguments are identical, meaning $$1+\frac12\cdot(\sqrt{4\beta+1}-1)=\frac2\beta\Leftrightarrow\beta=\frac{\sqrt[3]{3\cdot(9 + \sqrt{87})^2}-\sqrt[3]{36}}{3\cdot\sqrt[3]{9 + \sqrt{87}}}$$ and therefore $$\alpha=\frac2\beta=\frac{6\cdot\sqrt[3]{9 + \sqrt{87}}}{\sqrt[3]{6\cdot(9 + \sqrt{87})^2}-\sqrt[3]{36}}\approx1.696.$$ This shows the claim.
\end{proof}

\subsection{\hard{}-Bonus Model}
\label{subsec:general-intersection}

We now look at the general-evolution model with intersection bonus. This model is different from the ones considered before: We first give a simple lower bound showing that there is no constant-competitive algorithm.

\begin{theorem}\label{thm:general-intersection-lower}
   In the GE model with intersection bonus, there is no $c$-competitive algorithm for any constant $c$.
\end{theorem}
\begin{proof}
We consider an instance with no profit. Let $T=2$, $N=\{1,2\}$, and $\mathcal{F}_1=\{\emptyset,\{1\},\{2\}\}$, that is, there are two items, and at time $0$ it is only forbidden to take both of them. Assume w.l.o.g.\ that the algorithm does not pick item $2$ at time $1$. Then picking item $1$ becomes infeasible at time $2$ while picking item $1$ remains feasible. Then the algorithm achieves $0$ profit and bonus while the algorithm can achieve a bonus of $1$.
\end{proof}

Note that in this model, by adding dummy time steps giving no bonus and no profit, the previous lower bound extends to any number of time steps.
This lower bound motivates considering the 1-lookahead model: at time $t$, besides $I_t$, the algorithm knows the instance $I_{t+1}$. It shall decide the feasible solution chosen at time $t$. We consider an algorithm based on the following idea: at some time step $t$, the algorithm computes an optimal sequence of 2 solutions  $(S_{t,1}^*,S_{t,2}^*)$ of value $z^*_t$ for the subproblem defined on time steps $t$ and $t+1$. Suppose it fixes $S_t=S_{t,1}^*$. Then, at time $t+1$, it computes $(S_{t+1,1}^*,S_{t+1,2}^*)$ of value $z^*_{t+1}$. Depending on the values $z^*_t$ and $z^*_{t+1}$, it will either choose to  set $S_{t+1}=S^*_{t,2}$, confirming its choice at $t$ (getting in this case value $z^*_t$ for sure between time $t$ and $t+1$), or change its mind and set $S_{t+1}=S^*_{t+1,1}$ (possibly no value got yet, but a value $z^*_{t+1}$ if it confirms this choice at $t+2$). When a choice is confirmed ($S_t=S_{t,1}^*$ and $S_{t+1}=S_{t,2}^*$), then the algorithm starts a new sequence (fix $S_{t+2}=S^*_{t+2,1}$,\dots).

More formally, let $(S_{t,1}^*,S_{t,2}^*)$ be an optimal solution of the subproblem defined on time steps $t$ and $t+1$, and denote $z^*_t$ its value (including profits and bonus between time $t$ and $t+1$). To avoid unnecessary subcases, we consider at time $T$ $(S_{T,1}^*,S_{T,2}^*)$ where $S_{T,2}=\emptyset$ and $z^*_T$ is the profit of the optimal solution for the single time step $T$, $S_{T,1}^*$. Then consider the algorithm \texttt{Balance} which:

\begin{enumerate}
    \item At time $t=1$ compute $(S_{1,1}^*,S_{1,2}^*)$ and fix $S_1=S_{1,1}^*$.
    \item For $t=2$ to $T$: compute $(S_{t,1}^*,S_{t,2}^*)$.
    \begin{itemize}
        \item Case (1): If at $t-1$ the algorithm chose $S_{t-1}$ equal to $S_{t-2,2}^*$ (i.e., Case (3) occurred), then fix $S_t=S_{t,1}^*$.
        \item Case (2): Otherwise, if $z^*_t>2 z^*_{t-1}$, then fix $S_t=S^*_{t,1}$. 
        \item Case (3): Otherwise fix $S_t=S_{t-1,2}^*$.
    \end{itemize}
    
\end{enumerate}

\begin{theorem}\label{thm:general-intersection-upper}
In the GE model with intersection bonus and 1-lookahead, \texttt{Balance} is a 4-competitive algorithm.
\end{theorem}

\begin{proof}
Let $V$ be the set of time steps in which Case (3) occurred. In the proof, intuitively we partition the time period into periods which end at some time $t\in V$, and prove the claimed ratio in each of these sub-periods.

Formally, let $u,v$ ($u<v$) be two time steps in $V$ such that $w\not \in V$ for any $u<w<v$. Note that since Case (3) occurred at time $u$, Case (1) occurred at $u+1$, so $u\neq v-1$, and Case (2) occurred at time $u+2,\dots,v-1$. So $z^*_t>2z^*_{t-1}$ for $t=u+2,\dots,v-1$. By an easy recurrence, this means that, for all $t\in\{u+1,\dots,v-1\}$, we have $z^*_t< z^*_{v-1}/2^{v-1-t}$. By taking the sum, we get $\sum_{t=u+1}^{v-1}z^*_t < 2 z^*_{v-1}$. Since Case (3) occurred at $v$, $z^*_v\leq 2z^*_{v-1}$. Finally:

$$\sum_{t=u+1}^{v}z^*_t\leq 4 z^*_{v-1}.$$

Now, at each time $v$ for which case (3) occurred, we choose $S_v=S^*_{v-1,2}$. As previously said, Case (3) did not occur at $v-1$, so we choose $S_{v-1}=S^*_{v-1,1}$. Then the algorithm gets value at least $z^*_{v-1}$ for these two time steps. In other words
$f(S_1,\dots,S_T)\geq \sum_{v\in V} z^*_{v-1}$.
Consider first the case where $T\in V$ (case (3) occurred at time $T$). Then we get a partition of the time steps into subintervals ending in $v\in V$. So
$$\sum_{t=1}^{T}z^*_t\leq 4 \sum_{v\in V}z^*_{v-1}\leq 4 f(S_1,\dots,S_T).$$

Let $(\hat{S}_1,\dots,\hat{S}_T)$ be an optimal solution. We have $p_t(\hat{S}_t)+p_{t+1}(\hat{S}_{t+1})+b(\hat{S}_t,\hat{S}_{t+1})\leq z^*_t$. So $f(\hat{S}_1,\dots,\hat{S}_T)\leq \sum_{t=1}^{T-1}z^*_t$, and:
$$f(\hat{S}_1,\dots,\hat{S}_T)\leq  \sum_{t=1}^{T}z^*_t \leq 4 f(S_1,\dots,S_T).$$
Note that this is overestimated, each $p_t(\hat{S}_t)$ appears two times in the sum.

Now, if $T\not\in V$, then $T-1\in V$: indeed, Case (2) cannot occur at time $T$ (since $z^*_t\leq z^*_{t-1})$. So we have in this case:
$$\sum_{t=1}^{T-1}z^*_t\leq 4 \sum_{v\in V}z^*_{v-1}\leq 4 f(S_1,\dots,S_T).$$
But again since $f(\hat{S}_1,\dots,\hat{S}_T)\leq \sum_{t=1}^{T-1}z^*_t$, we have 
$$f(\hat{S}_1,\dots,\hat{S}_T)\leq  \sum_{t=1}^{T-1}z^*_t \leq 4 f(S_1,\dots,S_T).$$ This completes the proof.
\end{proof}

We prove a matching lower bound. The idea is as follows: As can be seen from the proof of Theorem~\ref{thm:general-intersection-upper}, the estimate on the profit has slack for the $4$-competitive algorithm. We give a construction in which there is no profit and in which the bonus when not ``committing'' to the solution from the previous time step is geometrically increasing over time; otherwise the bonus is $0$. As it turns out, however, when the factor is $2$ in each time step, we cannot show a lower bound of $4$ in case the algorithm does not commit until the last time step. Interestingly, if we use the minimum factor to show a lower bound of $4-\epsilon$ in case the algorithm commits at any time step but the last, we can find a large-enough time horizon such that, in case the algorithm commit only in the last time step, we can also show a lower bound of $4-\epsilon$.

\begin{theorem}\label{thm:general-intersection-lower1}
Consider the GE model with intersection bonus. For any $\epsilon >0$, there is a $T_\epsilon$ such that, for each number of time steps $T\geq T_\epsilon$, there is no $4-\epsilon$ competitive ratio.
\end{theorem}
\begin{proof}
Consider some $1$-lookahead algorithm \texttt{A} and $\epsilon\in(0,1)$. We will show that there is some number of time steps $T_\epsilon$ such that for all numbers of time steps $T\geq T_\epsilon$ \texttt{A} is not $(4-\epsilon)$-competitive. The construction is based on a sequence $a_1,a_2,\dots,a_{T_\epsilon}$ of natural numbers that we will determine later (along with $T_\epsilon$). In the ground set, there is precisely one item $(i,j)$ for each $i\in\mathbb{N}$ with $1\leq i\leq T_\epsilon$ and $j\in\mathbb{N}$ with $1\leq j\leq\max_{1\leq k\leq T_\epsilon}a_k$. We denote the set $\{(i,j)\mid 1\leq j\leq a_t\}$ by $R_{i,t}$. In this instance, value can only be obtained from transition bonuses.

Depending on the actions of \texttt{A}, we will define some time $t^\star$ with $2\leq t^\star\leq T_\epsilon$. At time $t>t^\star$, the empty set will be the only set that can be selected. At time $t=1,2$, selecting any set $R_{i,t}$ for some $i$ or the empty set is feasible. If \texttt{A} selects the empty set in either the first or the second time step, we simply set $t^\star:=2$.

Otherwise, at time $t$ with $2<t\leq t^\star$, selecting any set $R_{i,t}$ for some $i$ such that \texttt{A} has not selected $R_{i,t'}$ at any time $t'\leq t-2$ or the empty set is feasible. If \texttt{A} selects $R_{i,t-1}$ \emph{and} $R_{i,t}$ for some $i$ and $t\geq 2$, we say that \texttt{A} \emph{confirms} at time $t$. Then, or if \texttt{A} chooses the empty set at time $t$, set $t^\star:=t+1$; if \texttt{A} never confirms, then $t^\star:=T_\epsilon$. Note that this is a feasible construction for the $1$-lookahead model. 

We consider the competitive ratio in different cases:
\begin{itemize}
    \item If \texttt{A} chooses the empty set at some time $t\leq t^\star$, \texttt{A} does not obtain value at all while the optimum can obtain positive value (at least $a_1$), so \texttt{A} is not competitive.
    \item If \texttt{A} confirms at time $t\geq 2$, it obtains value $a_{t-1}$. Note that there exists some $i^\star$ so that \texttt{A} never chooses $R_{i^\star,t'}$ for any $t'$. The optimum chooses $R_{i^\star,t'}$ for all time steps $t'=1,\dots,\min\{t+1,T_\epsilon\}$. We distinguish two cases.
    \begin{itemize}
        \item We have $t+1\leq T_\epsilon$. Then the total value of the optimum is $\sum_{j=1}^{t}a_j$, leading to competitive ratio $\sum_{j=1}^{t}a_j/a_{t-1}$.
        \item We have $t+1>T_\epsilon$, implying $t=T_\epsilon$ (otherwise \texttt{A} could not have confirmed at $t$). Then the total value of the optimum is $\sum_{j=1}^{T_\epsilon-1}a_j$, leading to competitive ratio $\sum_{j=1}^{T_\epsilon-1}a_j/a_{T_\epsilon-1}$.
    \end{itemize}
    \item If \texttt{A} never confirms, \texttt{A} does not obtain value either while the optimum can obtain value $\sum_{j=1}^{T_\epsilon-1}a_j>0$. So in this case \texttt{A} is not competitive either.
\end{itemize}

For convenience, we will now describe a sequence $a_1',a_2',\dots,a_{T_\epsilon}'$ of rational numbers; $a_1,a_2,\dots,a_{T_\epsilon}$ can then be obtained by multiplying all numbers in the former sequence with a suitable natural number. Let $\epsilon'\in(0,\epsilon)$ be rational. Now the goal can be reformulated to be to choose $a_1',a_2',\dots,a_{T_\epsilon}'$ such that the ratios 
\begin{equation}\label{eq:kevin-not-last}
    \frac{\sum_{j=1}^{t}a_j}{a_{t-1}}=\frac{\sum_{j=1}^{t}a_j'}{a_{t-1}'}\text{ for all }t<T_\epsilon
\end{equation} and
\begin{equation}\label{eq:kevin-last}
    \frac{\sum_{j=1}^{T_\epsilon-1}a_j}{a_{T_\epsilon-1}}=\frac{\sum_{j=1}^{T_\epsilon-1}a_j'}{a_{T_\epsilon-1}'}
\end{equation} (corresponding to the above ones) are all at least $4-\epsilon'$. To do so, we start by setting $a_1':=1$. Now we inductively define $a_t'$ for $t\leq T_\epsilon-2$ (note $T_\epsilon$ is yet to be defined). Assuming all $a_1',\dots,a_{t-1}'$ are defined, we set $a_t'$ to be such that \eqref{eq:kevin-not-last} for $t$ is precisely $4-\epsilon'$. Equivalently, set $a_t':=(4-\epsilon')\cdot a_{t-1}'-\sum_{j=1}^{t-1} a_j'$.

We claim there exists a (first) $t_0$ such that \begin{equation}\label{eq:kevin2}
    (4-\epsilon')\cdot a_{t_0}'-\sum_{j=1}^{t_0} a_j'\leq a_{t_0}'
\end{equation} (meaning the $(t_0+1)$-st element of the sequence $a_1,a_2,\dots$ would become smaller than $a_{t_0}'$). Then we set $T_\epsilon:=t_0+2$ and $a_{t_0+2}'=a_{t_0+1}'=a_{t_0}'$. Note that then indeed, by~\eqref{eq:kevin2} and $a_{t_0+1}'=a_{t_0}'$,~\eqref{eq:kevin-not-last} is at least $4-\epsilon'$ for $t=T_\epsilon-1$ (and therefore, by the previous argument, for all). Furthermore, since $a_{t_0+2}'=a_{t_0+1}'$,~\eqref{eq:kevin-last} is identical to~\eqref{eq:kevin-not-last} for $t=T_\epsilon-1$ and therefore also at least $4-\epsilon'$.

So it remains to show the claim. Define $$b_{t}:=\frac{(4-\epsilon')\cdot a_{t}'-\sum_{j=1}^{t} a_j'}{a_{t}'}.$$ for all $t$ including the first element where the fraction becomes at most $1$ (if it exists; otherwise the sequence is infinite). Further note that for all such $t\geq 3$ we have $\sum_{j=1}^t a_j'=(4-\epsilon')\cdot a_{t-1}'$ (by using the definition for $a_{t-1}$). Therefore $b_{t}=(4-\epsilon')\cdot(1-a_{t-1}/a_t)=(4-\epsilon')\cdot(1-1/b_{t-1})$.

We now show two properties of the sequence $b_1,b_2,\dots$:
\begin{itemize}
    \item If $b_t\geq 1$ for $t\geq 3$, then $b_{t+1}<b_t$. Note that this expression simplifies to to $b_t^2-(4-\epsilon')\cdot b_t+(4-\epsilon')>0$, which is true for all $b_t\geq 1$.
    \item The sequence $b_1,b_2,\dots$ does not converge to a value at least $1$. Suppose it did. This would imply there exists $x\geq 1$ with $x=(4-\epsilon')\cdot(1-1/x)$, which however does not have a real solution.
\end{itemize}
By basic calculus, this proves that there exists a $t$ with $b_t\leq 1$, implying the claim, which in turn implies the theorem.
\end{proof}

\section{Conclusion}\label{sec:conclusion}

In this paper, we have developed techniques for online multistage subset maximization problems and thereby settled the achievable competitive ratios in the various settings almost exactly. Disregarding asymptotically vanishing terms in the upper bounds, what remains open is the exact ratio in the general-evolution setting with Hamming bonus (shown to be between $1+\sqrt{2}$ and $3$ in this paper) and exact bounds for the models with Hamming bonus when $T\rightarrow\infty$. Furthermore, it is plausible that the ratios can be improved for (classes of) more specific problems.

We emphasize that we have focussed on deterministic algorithms in this work. Indeed, some of our bounds can be improved by randomization (assuming an oblivious adversary):
\begin{itemize}
    \item In the general-evolution model with Hamming bonus assuming submodularity and subset feasibility, there is a simple randomized $(2+o(1))$-competitive algorithm (along the lines of the algorithms in Subsubsection~\ref{subsec:submodular}): Initially partition $N$ uniformly at random into two equal-sized sets (up to possibly one item) $A$ and $B$. At each time, select the optimal solution restricted to $A$. Again, the algorithm is $(2+o(1))$-competitive separately on both profit and bonus.
    \item While the strong lower bound without lookahead in the general-evolution model with intersection bonus still holds, we can get a simple $2$-competitive algorithm for lookahead $1$: Inititally flip a coin to interpret the instance as a sequence of length-$2$ instances either starting at time $1$ or $2$. Thanks to lookahead $1$, the length-2 instances can all be solved optimally. The total value of all these length-2 instances adds up to at least the optimal value, and the expected value obtained by the algorithm is half of that.
\end{itemize}

While we believe that we have treated various of the most natural ways of defining value in multistage subset maximization problems, other ways can be thought of, to some of which our results extend. For instance, Theorem~\ref{thm:static-hamming-upper} also works for time-dependent or object-dependent bonus without major modifications (whereas, e.g., Theorem~\ref{thm:static-intersection-upper} does not).

We have not worried about computational complexity in this work (and therefore neither about the representation of the set of feasible solutions); indeed, often we use an oracle providing the optimal solution to instances of a potentially hard problem. However, we mention that, if only an approximation algorithm to the problem at hand was known, we would be able to obtain similar online algorithms whose competitive ratio would depend on the approximation guarantee of the approximation algorithm.

\bibliography{biblio}
\end{document}